\documentclass[journal]{IEEEtran}

\usepackage{amsmath,amssymb,amsthm,bm,mathtools}
\usepackage{graphicx}
\usepackage[colorlinks=true,allcolors=blue]{hyperref}
\usepackage{booktabs}
\usepackage{xcolor}
\newtheorem{theorem}{Theorem}
\newtheorem{corollary}{Corollary}
\newtheorem{proposition}{Proposition}
\newtheorem{lemma}{Lemma}
\theoremstyle{definition}
\newtheorem{assumption}{Assumption}

\newtheorem{example}{Example}

\newcommand{\KL}[2]{D_{\mathrm{KL}}\!\left(#1\,\middle\|\,#2\right)}
\newcommand{\E}{\mathbb{E}}
\newcommand{\Nloc}{N_{\mathrm{local}}}
\newcommand{\Njoint}{N_{\mathrm{joint}}}
\newcommand{\Ph}{\Phi}
\newcommand{\rep}{\mathrm{rep}}
\newcommand{\tru}{\mathrm{joint}}
\newcommand{\diam}{\mathrm{diam}}
\newcommand{\Fcal}{\mathcal F}
\newcommand{\inc}{\rightarrow}
\newcommand{\Fst}[1]{\Fcal(#1)}
\newcommand{\Fr}[2]{\Fcal_{#1\inc#2}}
\newcommand{\dBcog}{\mathrm{dB_{cog}}}

\begin{document}


\title{6G Reasoning Agent Societies: Composition,\\
Scaling, and Consistency-Time Laws}

\author{
Hatim~Chergui,~\IEEEmembership{Senior~Member,~IEEE}, Carolina Fern\'{a}ndez-Mart\'{i}nez, Mustapha~Benjillali,~\IEEEmembership{Senior~Member,~IEEE,} Mehdi~Bennis,~\IEEEmembership{Fellow,~IEEE},  and Merouane~Debbah,~\IEEEmembership{Fellow,~IEEE}

\IEEEcompsocitemizethanks{\IEEEcompsocthanksitem H. Chergui and C. Fern\'{a}ndez-Mart\'{i}nez are with the i2CAT Foundation, Spain. (e-mails: chergui@ieee.org, name.surname@i2cat.net)}
\IEEEcompsocitemizethanks{\IEEEcompsocthanksitem M.~Benjillali is with INPT, Morocco.  (e-mail: benjillali@ieee.org).}
\IEEEcompsocitemizethanks{\IEEEcompsocthanksitem M. Bennis is with the University of Oulu, Finland (e-mail: mehdi.bennis@oulu.fi).}
\IEEEcompsocitemizethanks{\IEEEcompsocthanksitem M. Debbah is with the Research Institute for Digital Future, Khalifa University, 127788 Abu Dhabi, UAE (e-mail: merouane.debbah@ku.ac.ae).}
\IEEEcompsocitemizethanks{\IEEEcompsocthanksitem This work was supported by grants DGR-2026-2-DIGITALSHIELD, funded by the Research and Universities Department of the Catalonian Government as well as by COALESCE-6G PID2024-163028OB-I00, funded by MICIU/AEI/10.13039/501100011033/FEDER, EU.}
}

\maketitle

\begin{abstract}
Emerging 6G control planes would compose Large Language Model (LLM) agents that exchange interpreted conclusions, so hallucinations can propagate even over error-free transport. Recent architectures audit each message locally against the sender's observation mask, but isolated audits miss collective behavior. Modeling the agent society as a cellular sheaf, we establish three laws valid on arbitrary topologies. First, the \emph{composition law}: network-wide divergence exceeds the sum of per-agent audit scores by exactly their hidden total correlation. This blind spot masks up to 23\% of joint divergence and, because the declared signal cancels in the cognitive signal-to-noise ratio (SNR), overstates it by $1.14$\,dB purely through topology; the resulting trust inflation preserves valid consensus states but slows convergence. Second, the \emph{scaling law}: reconciliation decouples across metric counters, so the worst-connected dimension sets the pace. Consistency time grows for a society size $n$ as $\Omega(n^{2/d})$ on polynomial-growth topologies of degree $d$, and paths breach a $1$\,s near-real-time budget at $n^\ast=13$ agents, whereas expanders ($\Theta(1)$) and radius-3 local audits ($O(1)$, $\approx300$\,ms) stay bounded. Third, the \emph{consistency-time law}: the timing error from omitted correlations is one-sided and capped by the worst-measured link. These laws, validated on simulated topologies and deployed LLM societies, give provable worst-case latency ceilings for near-real-time operation.
\end{abstract}


\section{Introduction}

\IEEEPARstart{P}{roposals} for next-generation O-RAN architectures are
rethinking the near-real-time RAN Intelligent Controller (RIC). Instead of
rigid, rule-based xApps, these networks employ LLM agents that do more than
execute code: they interpret E2 measurement reports, form internal beliefs
about the state of the network's cells, and negotiate resource allocation in
natural language before acting~\cite{ref:tmforum,ref:bariah}. This leap in
capability introduces a blind spot. Because the agents are proprietary black
boxes, an operator cannot audit their internal logic to verify how a
conclusion was reached, and the control plane is no longer transmitting bits
reliably but conveying interpreted \emph{meaning}~\cite{ref:strinati,ref:park}.
Protection has to be redefined accordingly. In a classical control plane,
resilience is transport integrity: a packet delivered intact is correct. In an
agentic one, a perfectly delivered conclusion can still be a hallucination
inferred from misunderstood or unmonitored metrics, so the network must secure
\emph{trust} rather than integrity~\cite{ref:tom,ref:saad}, which means
endowing it with a machine Theory of Mind---the ability to judge an agent's
credibility from its behavior and interactions~\cite{ref:kosinski,ref:ullman}.

\subsection{Literature Review}

Four lines of work bear on that problem and none of them closes it.

The first is the agentic control plane itself. The near-real-time RIC, its E2
interface and the xApp model that LLM agents are proposed to inhabit are well
characterized~\cite{ref:polese}, and multi-agent frameworks show that agents
given roles, memories and a protocol do coordinate on non-trivial
tasks~\cite{ref:tmforum,ref:bariah,ref:park}. The same literature establishes
the failure mode that motivates auditing: a language model produces fluent,
confident statements unsupported by its inputs~\cite{ref:ji}, and at the
receiver those are indistinguishable from well-founded ones. Outputs are
nevertheless judged one at a time, which is the assumption this paper prices.

The second is semantic and goal-oriented communication, which moves the design
objective from bit fidelity to the fidelity of conveyed meaning and gives an
information-theoretic account of what a receiver needs for its
task~\cite{ref:strinati,ref:saad,ref:gunduz}. It says what should be preserved
end to end, not what an \emph{evaluator} of meaning can and cannot see---a
different question, and the one a trust layer faces.

The third is trust and reputation. Scores are accumulated per interaction and
aggregated by sums, averages or Bayesian updates~\cite{ref:josang}, and in
agentic control planes each message is scored against an expectation
reconstructed from the sender's public observation mask~\cite{ref:tom}. The
locality is not an oversight but a necessity, since a near-real-time loop
cannot run joint inference over message histories. Every one of these
aggregation rules, however, is exact only when the evidence being combined is
independent---and circulating conclusions introduce dependence.

Fourth, the spectral theory of agreement models network consensus as a linear diffusion process, $\dot{x} = -Lx$, where nodes iteratively average their states with neighbors~\cite{ref:olfati,ref:chung,ref:mohar}. The rate of this diffusion follows a $1/\lambda_2$ relaxation law, meaning the graph Laplacian's spectral gap $\lambda_2$ controls how fast differences dissipate. Consequently, the edge weights of $L$ can be treated as tunable design parameters and optimized directly to maximize convergence speed~\cite{ref:xiao}. Cellular sheaves generalize this to agents holding
different quantities that agree only where their vocabularies
overlap~\cite{ref:hansen}, and sheaf diffusion has modeled disagreement among
agents whose private frames differ~\cite{ref:hansenop}. Classical information
theory supplies the measures: Kullback--Leibler
geometry~\cite{ref:csiszar}, total correlation as the multivariate extension
of mutual information~\cite{ref:watanabe}, and the composition results
separating single-link from cascaded capacity~\cite{ref:cover,ref:covergamal}.

\subsection{Gaps and Contributions}

Three gaps remain. First, the composition results of network information
theory govern \emph{transport}: what a cascade of channels can carry, not what
a cascade of \emph{evaluations} can detect. With no counterpart for
credibility, a trust layer that sums per-link scores inherits the sum's
independence assumption, and with it a cost no better expectation can
remove.

Second, consensus theory prices agreement for agents holding the same state
but disagreeing about its value. Sheaf
Laplacians~\cite{ref:hansen,ref:hansenop} handle different vocabularies, but
ask whether opinions converge, and to what---not how long that takes under the
topology the subscriptions induce, which is what a near-real-time budget is
spent on.

Third, the weights here are measurements, not design
variables. Where~\cite{ref:xiao} chooses weights to make a network fast, a
trust layer reads them from scores that Theorem~\ref{thm:main} in this paper shows to be
optimistic, then derives its convergence time from them. A controller
sizing a budget on that number trusts a self-report whose error the
literature does not bound.

We close these gaps with three laws, each stated for an arbitrary topology,
proved from the architecture and validated without fitted parameters:

\begin{itemize}
\item \emph{The composition of trust} (Section~\ref{sec:identity}). Link-by-link
evaluation underestimates the society's divergence by exactly the total
correlation of its messages, for any topology and per-agent expectations.
\item \emph{Scaling and consistency time} (Section~\ref{sec:scaling}). A
society reconciles as slowly as its worst-connected counter, so topology alone
fixes consistency time: $\Omega(n^{2/d})$ under polynomial growth, $\Theta(1)$
on expanders, and $O(1)$ within a fixed neighborhood.
\item \emph{Bounding the reported consistency time}
(Section~\ref{sec:scaling}). The reported consistency time never exceeds the
true one and falls short of it by at most the worst-audited link's inflation,
so near-real-time budgets can be sized on it safely.
\item \emph{Validation} (Section~\ref{sec:results}). The laws are tested on
simulated sheaves swept over topology and size, and on LLM agent societies
prompted with O-RAN measurement reports.
\end{itemize}
\section{System Model}\label{sec:model}

\begin{figure}[tb]
  \centering
  \includegraphics[width=\linewidth]{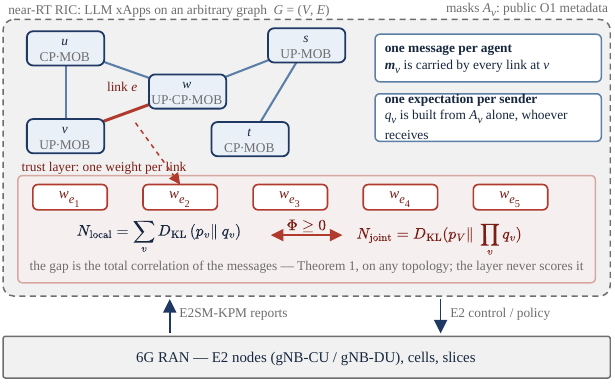}
  \caption{A society of LLM xApps on a near-real-time RIC. The agents interact over graph $G=(V,E)$, each monitoring the RAN through a public mask $A_v$. The trust layer scores each link against an expectation built from the sender's mask alone.}
  \label{fig:sys}
\end{figure}

\subsection{Agents, Masks, Messages, and Beliefs}\label{ssec:recall}

The control plane operates as a society of agents connected over a graph
$G=(V,E)$, shown in Fig.~\ref{fig:sys}; $\mathcal N(v)$ denotes the neighbors
of agent $v$. Each agent subscribes to a specific set of E2 counters, forming
its public observation mask $A_v$, and observes the network metrics
$o_v=\bm\theta|_{A_v}$ for a given cell state $\bm\theta$. Agents score the
cell over a shared vocabulary $\mathcal C$ of scoring dimensions $c$---in the
running example, user-plane load (\texttt{UP}), control-plane signaling
(\texttt{CP}), and mobility stress~(\texttt{MOB})---each fed by one or more E2
counters; $\mathcal C_v\subseteq\mathcal C$ denotes the dimensions covered by
$A_v$.

Each round, agent $v$ emits a message over the whole vocabulary, produced by
its proprietary policy $\pi_v$ from its own observation and the messages its
neighbors sent in the previous round. Its \emph{belief} is the part of that
message its mask backs:
\begin{equation}
\mathbf m_v=\pi_v\big(o_v,\{\mathbf m_u\}_{u\in\mathcal N(v)}\big)\in[0,100]^{\mathcal C},
\qquad
\mathbf x_v=\mathbf m_v\big|_{\mathcal C_v}.
\label{eq:msg}
\end{equation}
Each entry of $\mathbf x_v$ is a score on the common $0$--$100$ scale, the mean
of a belief distribution with variance $\sigma^2_{v,c}$ on dimension~$c$; the
belief's precision matrix is $\Lambda_v=\mathrm{diag}(\sigma^{-2}_{v,c})_{c\in\mathcal C_v}$.
The remainder of the message, $\mathbf m_v|_{\mathcal C\setminus\mathcal C_v}$,
is asserted without evidence. The message is indexed by vertex because the
same vector traverses every incident link. Because masks are public O1
metadata, an operator can invert a generic emission policy to calculate an
expected belief $\hat{\mathbf x}_v$ on the same dimensions, based strictly on
what $v$ is allowed to see---its observation $o_v$, with no neighbor input.

Crucially, asking agents with asymmetric views to agree on everything is
ill-posed. If a signaling-focused agent lacks user-plane counters and a
user-plane agent lacks control-plane counters, their beliefs do not even range
over the same dimensions, and forcing $\mathbf x_u=\mathbf x_v$ would require
both to guess blindly on unmonitored metrics. Consensus is meaningful only on
the \emph{projection} onto mutually observable metrics. The mathematical
framework that formalizes this kind of localized, overlapping agreement is a
cellular sheaf~\cite{ref:hansen}.

\subsection{A Running Example}\label{ssec:example}

To illustrate why partial agreement is necessary and how the trust layer
evaluates it, consider the following scenario.

\begin{example}[Two agents and a signaling storm]\label{ex:two}
Two agents monitor a single cell: \emph{Signaling Watch} ($u$, for the control plane) and \emph{Slice Guardian} ($v$, for the user plane); both also monitor mobility. Their masks $A_u$ and $A_v$ capture data during a signaling storm occurring over a lightly loaded user plane.

\begin{center}\footnotesize
\begin{tabular}{@{}lllr@{}}
\toprule
Mask & E2SM-KPM counter & Dimension & Reading \\
\midrule
$A_u$ & \texttt{RACH.PreambleDedCell} & \texttt{CP} & 2740 \\
      & \texttt{MM.HoExeFailCell}     & \texttt{MOB} & 46 \\
\addlinespace[1pt]
$A_v$ & \texttt{RRU.PrbTotDl}         & \texttt{UP} & 32 \\
      & \texttt{DRB.UEThpDl}          & \texttt{UP} & 51 \\
      & \texttt{MM.HoExeFailCell}     & \texttt{MOB} & 46 \\
\bottomrule
\end{tabular}
\end{center}

Agent $u$ emits the message $\mathbf m_u=(70,85,55)$ over
$(\texttt{UP},\texttt{CP},\texttt{MOB})$. Agent $v$, after reading it, emits
$\mathbf m_v=(52,60,58)$, even though its own counters would have put the user
plane near $35$.
\end{example}

By~\eqref{eq:msg}, $\mathbf m_u$ restricts to the belief
$\mathbf x_u=(-,85,55)$, and $\mathbf m_v$ to $\mathbf x_v=(52,-,58)$. The
message $\mathbf m_u$ is a single point; \emph{the law} $p_u$ is the
statistical distribution this message follows across operating
conditions---it should be sharp on \texttt{CP} but diffuse on \texttt{UP}---and
\emph{the expectation law} $q_u$ is what the trust layer builds from the mask
alone.

Agent $u$'s \texttt{UP} score of $70$ is a confident assertion its mask cannot
support, and the trust layer rightfully penalizes it. Agent~$v$'s \texttt{UP}
score of $52$, by contrast, passes the audit because its mask covers
\texttt{UP}---yet that score was heavily influenced by reading $u$'s
hallucinated message. No local inspection against $q_v$ can detect this flaw:
the provenance of the error is hidden entirely in the statistical dependence
between the agents.

\subsection{The Trust Sheaf}\label{ssec:sheaf}

To manage these overlapping views mathematically, we model the network as a
trust sheaf. Extending the example, $u$ still tracks \texttt{CP} and
\texttt{MOB}, $v$ tracks \texttt{UP} and \texttt{MOB}, and a third agent $z$
tracks all three dimensions; in this configuration $u$ and $v$ are each linked
to $z$ but not to each other.

\emph{Stalks and shared vocabularies.} A sheaf $\Fcal$ assigns each agent a
\emph{vertex stalk} $\Fst{v}$, the space of its beliefs over $\mathcal C_v$,
and each link an \emph{edge stalk} $\Fst{e}$, the dimensions its two endpoints
share. Before communicating across a link, an agent applies a
\emph{restriction map}
\begin{equation}
\Fr{v}{e}:\;\Fst{v}\longrightarrow\Fst{e},\qquad v\inc e,
\label{eq:restr}
\end{equation}
which acts as a filter, dropping every dimension its neighbor does not track.

\emph{Defining consensus.} Consensus on a link $e=\{u,z\}$ means the agents'
filtered beliefs match: $\Fr{u}{e}\mathbf x_u=\Fr{z}{e}\mathbf x_z$. If $u$'s
belief is $(80,57)$ and $z$'s is $(44,80,57)$, they agree on their link:
$u$'s $(80,57)$ equals $z$'s restriction to the shared \texttt{CP} and
\texttt{MOB}, and $u$'s blindness to $z$'s \texttt{UP} score of $44$ does not
break the agreement.

\emph{Measuring disagreement.} When agents disagree, their current beliefs are
stacked into a global state vector, a \emph{$0$-cochain}
$\mathbf x\in\bigoplus_{v}\Fst{v}$. For the three agents, with $u$ and $v$
holding their beliefs from Example~\ref{ex:two}:
\begin{equation}
\mathbf x=\big(\underbrace{-,85,55}_{u}\;\big|\;\underbrace{52,-,58}_{v}\;\big|\;
\underbrace{44,80,57}_{z}\big).
\label{eq:cochain}
\end{equation}
The \emph{coboundary operator} $\bm\delta$ subtracts the restricted beliefs
across each link. Between $u$ and $z$ (sharing \texttt{CP}, \texttt{MOB}) the
disagreement is $(85-80,\,55-57)=(5,-2)$; between~$v$ and~$z$ (sharing
\texttt{UP}, \texttt{MOB}) it is $(52-44,\,58-57)=(8,1)$. The resulting
\emph{$1$-cochain} collects these residuals:
\begin{equation}
\bm\delta\mathbf x=(5,-2\mid 8,1).
\end{equation}

\emph{The Laplacian and relaxation.} The network penalizes these residuals in
proportion to the trust weight $w_e$ of each link, defined in
Section~\ref{ssec:score}. With $\mathbf W$ the diagonal matrix of link
weights, the \emph{sheaf Laplacian} $\mathbf L_{\Fcal}\triangleq\bm\delta^\top\mathbf W\bm\delta$
yields their weighted sum:
\begin{equation}
\mathbf x^\top\mathbf L_{\Fcal}\mathbf x=\sum_{e}w_{e}\|(\bm\delta\mathbf x)_e\|^2 .
\label{eq:lap}
\end{equation}
To resolve conflicts, the network undergoes \emph{relaxation},
$\dot{\mathbf x}=-\mathbf L_{\Fcal}\mathbf x$: a local gradient descent that
moves each agent's beliefs toward its trusted neighbors, and stops only when
every shared dimension agrees and $\bm\delta\mathbf x=0$.

\emph{Gluing.} Gluing is the property that lets local agreements cascade into
network-wide consensus. An assignment of beliefs is \emph{consistent} when all
differences across links vanish, $\bm\delta\mathbf x=0$; such a state is a
\emph{global section}, an element of $H^0(G;\Fcal)=\ker\bm\delta$. Anchoring
$z$'s beliefs, an audit of the two links to $z$ alone drives
$\mathbf x_u=(80,57)$ and $\mathbf x_v=(44,57)$, and so brings $u$ and $v$,
which share no link, to a common $\texttt{MOB}=57$. On a connected graph,
pairwise agreement stitches itself into global agreement, which is what
licenses a purely local audit. Because $\mathbf W$ is positive diagonal,
$\ker\mathbf L_{\Fcal}=\ker\bm\delta=H^0$ whatever the weights: trust sets how
fast the society reaches a global section, never which global sections exist.

If calibrations conflict around loops, consensus contracts to the consistent
subspace $H^0$, and relaxation projects the initial beliefs onto it with
disagreement decaying as $e^{-\lambda_2 t}$~\cite{ref:olfati}, where
$\lambda_2$ is the smallest nonzero eigenvalue of $\mathbf L_{\Fcal}$. The
time for disagreement to fall to a fraction $\varepsilon$ of its initial value
is the \emph{consistency time}
\begin{equation}
\tau(\Fcal)\triangleq\lambda_2(\mathbf L_{\Fcal})^{-1}\log(1/\varepsilon).
\label{eq:tau}
\end{equation}

\subsection{The Cognitive SNR}\label{ssec:score}

Receivers evaluate trust via a machine Theory of Mind, auditing whether a
sender's evidence supports its claims. Comparing a transmitted message against
an expectation derived from the sender's public mask yields a \emph{cognitive
signal-to-noise ratio (SNR)}~\cite{ref:tom}:
\begin{equation}
\begin{aligned}
N_{v\inc e}&\triangleq\KL{\Fr{v}{e}\,\mathbf x_v}{\Fr{v}{e}\,\hat{\mathbf x}_v},
&\gamma_{v\inc e}&\triangleq\frac{S_v}{N_{v\inc e}},\\[2pt]
w(\gamma)&\triangleq1-\exp({-\eta\gamma})\in[0,1),
\end{aligned}
\label{eq:csnr}
\end{equation}
where, as throughout, a divergence is taken between the belief distributions
whose means these vectors hold. The \emph{noise} $N_{v\inc e}$ penalizes
hallucinations by measuring how far the sender's restricted message departs
from the expectation its mask supports. The \emph{signal} $S_v$ is the
confidence the sender asserts, taken as in the companion
formulation~\cite{ref:tom} to be the trace of its belief's precision matrix:
\begin{equation}
S_v\triangleq\operatorname{Tr}\!\big(\Lambda_v\big)
=\sum_{c\in\mathcal C_v}\frac{1}{\sigma^2_{v,c}} .
\label{eq:signal}
\end{equation}
Since scores are pure numbers on the $0$--$100$ scale, $S_v$ is a
dimensionless precision, whereas $N_{v\inc e}$ is a divergence in bits; the
ratio $\gamma_{v\inc e}$ therefore carries units of inverse bits, reading as
asserted precision per bit of unexplained divergence. It is reported against
a declared reference $\gamma_0=1\,\mathrm{bit}^{-1}$, exactly as power is
reported in dBm against $1$\,mW, and written in $\dBcog$: a value of
$x\,\dBcog$ means $10\log_{10}(\gamma_{v\inc e}/\gamma_0)=x$. Nothing in the
three laws below turns on that choice: each is stated as a \emph{ratio} in
which $S_v$ cancels, so the reference cancels with it.

The ratio $\gamma_{v\inc e}$ acts as a native hallucination detector:
asserting high confidence without observational support spikes
$N_{v\inc e}$ and collapses the SNR. Finally, $\eta$ scales $\gamma$ into a
normalized weight, and reconciliation across a link is bottlenecked by its
least credible participant: $w_e\triangleq w(\gamma_e)$ with
$\gamma_e\triangleq\min_{v\inc e}\gamma_{v\inc e}$.

\subsection{Message Laws and the Two Divergences}\label{ssec:two}

Across operating conditions $\bm\theta$, \eqref{eq:msg} makes each
$\mathbf m_v$ a random vector. The ensemble
$\mathbf m_V\triangleq(\mathbf m_v)_{v\in V}$ follows a joint law $p_V$ with
marginals $p_v$, and its dependence has two sources: shared counters feed a
common $\bm\theta$ into several observations, and \eqref{eq:msg} feeds each
agent its neighbors' messages. The trust layer's expectation $q_v$ is the law
of the message a generic policy would emit from $o_v$ alone, with no neighbor
input. Agent $v$ contributes $N_v\triangleq\KL{p_v}{q_v}$, taken over the
whole vocabulary so that unbacked claims are charged, and the reported and
joint divergences are
\begin{equation}
\Nloc\triangleq\sum_{v\in V}N_v,
\qquad
\Njoint\triangleq\KL{p_V}{\textstyle\prod_{v\in V}q_v}.
\label{eq:nloc}
\end{equation}
Both measure cognitive noise in bits. In Example~\ref{ex:two}, $\Nloc$ charges
$u$ for its unbacked \texttt{UP} claim and clears $v$, whose \texttt{UP} score
is backed by its own counters; $\Njoint$ also charges for $v$'s \texttt{UP}
score moving with $u$'s, a dependence no marginal can represent.

\begin{assumption}[Per-agent expectation]\label{as:main}
The reference law against which the trust layer scores a relayed conclusion is
the product $\prod_{v\in V}q_v$, each factor built from its agent's own
observation alone.
\end{assumption}

\section{The Composition Law}\label{sec:identity}

Near-real-time control loops must evaluate messages link by link, because
joint audits over full network histories are computationally infeasible. Yet
as agents read peer messages, their conclusions become statistically
intertwined. Local audits measure individual deviations from expectation but
cannot detect collective reinforcement. The composition law formalizes this
blind spot: localized auditing underestimates the joint divergence by exactly
the hidden dependence among agents, and so inflates reported trust.

\begin{theorem}[Society divergence]\label{thm:main}
For any finite society and any interaction topology, the locally visible and
joint divergences of~\eqref{eq:nloc} differ by exactly the total correlation
$\Ph$ of the messages $\mathbf m_V$:
\begin{equation}
\Njoint=\Ph(\mathbf m_V)+\Nloc ,
\label{eq:main}
\end{equation}
where $\Ph(\mathbf m_V)\triangleq\KL{p_V}{\prod_{v\in V}p_v}
=\sum_{v\in V}h(\mathbf m_v)-h(\mathbf m_V)$ and $h(\cdot)$ is the
differential entropy.
\end{theorem}

\begin{IEEEproof}
Assume $p_V\ll\prod_v q_v$ and all terms finite. This holds in practice
because each $q_v$, being the law of a generic policy's output, assigns
positive density to every score on every dimension---on dimensions outside the
mask it reduces to an uninformative prior---so no feasible message is declared
impossible. Inserting $\log\prod_v p_v$,
\begin{align}
\Njoint&\triangleq\E_p\!\left[\log\frac{p_{V}}{\prod_v q_v}\right] \nonumber\\
&=\underbrace{\E_p\!\left[\log\frac{p_{V}}{\prod_v p_v}\right]}_{(a)}
 \!+\!\underbrace{\E_p\!\left[\log\frac{\prod_v p_v}{\prod_v q_v}\right]}_{(b)}\!.
\label{eq:three}
\end{align}
Term $(a)$ is $\KL{p_V}{\prod_{v\in V}p_v}=\Ph(\mathbf m_V)$. By linearity of
expectation, term $(b)$ splits into
$\sum_{v\in V}\E_p[\log p_v(\mathbf m_v)/q_v(\mathbf m_v)]=\Nloc$.
\end{IEEEproof}

The proof uses only Assumption~\ref{as:main} and linearity of
expectation---no ordering, no Markov property---so the topology shapes $p_V$
without entering the statement, and independent messages give $\Ph=0$. In
Example~\ref{ex:two} the shortfall is $\Ph(\mathbf m_{\{u,v\}})=I(\mathbf m_u;\mathbf m_v)$,
which reflects both the shared \texttt{MOB} counter and $v$'s persuasion.
Three readings follow at once. Since $\Ph(\mathbf m_V)\ge0$ by Gibbs'
inequality, $\Njoint\ge\Nloc$, with equality iff the messages are mutually
independent, so a per-link audit is optimistic and never pessimistic. The
shortfall does not depend on the per-agent expectations: term $(a)$ involves
only the $p_v$, never the $q_v$, so $\Ph(\mathbf m_V)$ survives any choice of
$q_v$, and no amount of care in modeling the senders recovers it. And because
$\gamma$ decreases with divergence while $w$ increases with $\gamma$, the gap
reaches the reported trust weights intact, rescaling the society's confidence
by a single factor.

\begin{corollary}[Uniform SNR inflation]\label{cor:beta}
Let $S\triangleq\sum_{v\in V}S_v$. The reported and joint-consistent cognitive
SNRs $\gamma^{\rep}=S/\Nloc$ and $\gamma^{\tru}=S/\Njoint$ satisfy
\begin{equation}
\frac{\gamma^{\rep}}{\gamma^{\tru}}
=\frac{\Njoint}{\Nloc}
=1+\frac{\Ph}{\Nloc}\;\triangleq\;\beta\;\ge\;1 .
\label{eq:beta}
\end{equation}
\end{corollary}

\begin{IEEEproof}
By~\eqref{eq:signal}, $S_v$ depends only on what $v$ transmitted, not on any
reference law, so it is objective and receiver-computable yet intrinsic to the
sender. The local evaluation sums these per-agent signals. The joint
evaluation reads the society's asserted precision, which is block-diagonal
because no agent transmits cross-covariances; its trace is therefore the same
$\sum_{v}S_v=S$. The two evaluations thus share the confidence $S$ and differ
only in their reference law, so $S$ cancels and~\eqref{eq:beta} is
Theorem~\ref{thm:main} divided by $\Nloc$. Since $\Ph(\mathbf m_V)\ge0$,
$\beta\ge1$.
\end{IEEEproof}

This is what lets the shortfall be read as a power ratio without any
commitment about how confidence is measured. Both evaluations share the same
declared signal $S$, which cancels identically, so $\beta$ is a dimensionless
factor by which the audit over-reports the society's cognitive SNR. Its
decibel reading, $10\log_{10}\beta$, is therefore an unqualified decibel,
produced by the topology and the message correlations rather than by the
agents' calibration, and no recalibration can reduce it.

\section{Scaling and Consistency Time}\label{sec:scaling}

When a network scales, smaller local societies are glued into larger ones, as
described in Section~\ref{ssec:sheaf}. Gluing preserves locally certified
states, and the consensus space $H^0=\ker\bm\delta$ does not depend on the
weights; so neither growth nor the trust inflation of
Corollary~\ref{cor:beta} alters which agreements are valid. What both change
is the \emph{rate} of convergence---\emph{how fast} agents agree, not
\emph{whether} they agree. To understand this speed limit, consider how the
structure tracks variables. Because links compare only shared
counters---comparing \texttt{CP} while ignoring \texttt{UP}, for
instance---dimensions never mix, and the Laplacian decouples into independent
reconciliation processes, one per dimension.

\begin{lemma}[The society splits by dimension]\label{lem:gap}
Let every restriction map be a coordinate projection onto the dimensions its
link shares, and let $G_c$, the subgraph sharing dimension $c$, be connected.
The Laplacian decomposes as $\mathbf L_{\Fcal}=\bigoplus_{c}\mathbf L_{G_c}^{w}$,
where $\mathbf L_{G_c}^{w}$ is the graph Laplacian of $G_c$ carrying the trust
weights. The algebraic connectivity and the consistency time are therefore set
entirely by the slowest dimension:
\begin{equation}
\lambda_2(\mathbf L_{\Fcal})=\min_{c}\lambda_2(\mathbf L_{G_c}^{w}),
\qquad
\tau(\Fcal)=\max_{c}\tau(G_c).
\label{eq:gaplemma}
\end{equation}
\end{lemma}

\begin{IEEEproof}
The Laplacian~\eqref{eq:lap} computes differences across links with no
cross-terms between dimensions. The spectrum of the direct sum is the union of
the individual spectra, with unshared dimensions contributing only to the
kernel.
\end{IEEEproof}

This creates a counter-intuitive dynamic: widely shared counters demand global
consensus and become bottlenecks, while narrowly shared dimensions reconcile
quickly across local links. Because Lemma~\ref{lem:gap} isolates each
dimension into an ordinary graph $G_c$, classical graph theory dictates how
topology governs the worst-case consistency time. Below, $B_r(v)$ denotes the
radius-$r$ ball around $v$, the agents within $r$ hops of it, and a graph has
polynomial growth $d$ when its ball sizes scale as $r^d$.

\begin{proposition}[Consistency time vs.\ topology]\label{prop:scaling}
Let the consistency time~\eqref{eq:tau} be evaluated for link weights
$w\ge w_{\min}>0$, on any sheaf $\Fcal'$ over a subgraph of $G$. Topology
dictates convergence:
\emph{(i) Global scaling:} On graphs with polynomial growth $d$ (e.g., grids)
and $n$ agents, consistency slows as $\Omega(n^{2/d})\log(1/\varepsilon)$.
Expander graphs, conversely, maintain a constant time
$\Theta(1)\log(1/\varepsilon)$~\cite[Ch.~1--2]{ref:chung}.
\emph{(ii) Local scaling:} Consistency within any ball $B_r(v)$ depends only on
local geometry, independent of $n$:
\begin{equation}
\tau\big(\Fcal|_{B_r(v)}\big)\;\le\;
\frac{|B_r(v)|\,\diam(B_r(v))}{4\,w_{\min}}
\,\log\frac{1}{\varepsilon}.
\label{eq:scalinglocal}
\end{equation}
\emph{(iii) Universal bound:} If the maximum degree $\Delta$ is bounded, every
radius-$r$ ball reconciles in $O(1)$ time with respect to $n$, whatever the
global topology.
\end{proposition}

\begin{IEEEproof}
By Lemma~\ref{lem:gap}, it suffices to analyze each $G_c$, the slowest of
which sets the overall speed.
\emph{(i)} On polynomial-growth graphs, standard spectral theory gives
$\lambda_2=O(n^{-2/d})$, which through~\eqref{eq:gaplemma} yields the
$\Omega(n^{2/d})$ bound. Expander graphs preserve a constant gap
$\Theta(1)$~\cite[Ch.~1--2]{ref:chung}, so the society is fast if every
dimension expands.
\emph{(ii)} Restricting each $G_c$ to $B_r(v)$ gives a subgraph $G'$ with at
most $|B_r(v)|$ vertices and diameter at most $\diam(B_r(v))$. Mohar's lower
bound on algebraic connectivity~\cite{ref:mohar} states that for any graph
$G'=(V',E')$ with minimum edge weight $w_{\min}$,
$\lambda_2\ge 4w_{\min}/(|V'|\,\diam(G'))$. Inverting $\lambda_2$
yields~\eqref{eq:scalinglocal}, which depends only on $r$ and the maximum
degree.
\emph{(iii)} With maximum degree $\Delta$, the Moore bound gives
$|B_r(v)|\le 1+\Delta\sum_{k<r}(\Delta-1)^k$~\cite{Miller_Siran_2005}, and
$\diam(B_r(v))\le 2r$, so~\eqref{eq:scalinglocal} is $O(1)$ in $n$.
\end{IEEEproof}

Table~\ref{tab:scaling} summarizes this architectural reality: standard
physical topologies such as paths and grids cannot achieve global consensus at
scale, and only expanders or dense networks maintain a bounded global
consistency time. Otherwise, architects must rely on~\eqref{eq:scalinglocal},
ensuring that well-connected local neighborhoods reconcile quickly,
independent of the broader network.

\begin{table}[tb]
\centering\footnotesize
\caption{Consistency time versus topology as $n$ grows. For these named
families the order is tight, so the $\Omega(n^{2/d})$ of
Proposition~\ref{prop:scaling} sharpens to $\Theta(\cdot)$. The last column
is the exponent of $n$ fitted to the measured consistency time.}
\label{tab:scaling}
\begin{tabular}{@{}llll r@{}}
\toprule
Topology & $d$ & $\lambda_2(\mathbf L_G)$ & $\tau(\Fcal)$ & Measured \\
\midrule
Path, Ring & $1$      & $\Theta(n^{-2})$ & $\Theta(n^{2})$   & $2.00$, $1.99$ \\
Grid       & $2$      & $\Theta(n^{-1})$ & $\Theta(n)$       & $0.98$ \\
Expander   & $\infty$ & $\Theta(1)$      & $\Theta(1)$       & $-0.12$ \\
Complete   & ---      & $\Theta(n)$      & $\Theta(n^{-1})$  & $-1.03$ \\
\bottomrule
\end{tabular}
\end{table}

While Proposition~\ref{prop:scaling} maps theoretical speeds, real-world
auditing is imperfect. Applying Theorem~\ref{thm:main} to the sub-society each
link's audit sees gives an inflation $\beta_e\ge1$ per link $e$
(Corollary~\ref{cor:beta}), so the network is driven by distorted weights. The
$\beta_e$ need not be equal across links, and Theorem~\ref{thm:bias} uses only
$\beta_e\ge1$, not their values. This creates a discrepancy between the
\emph{reported} consistency time $\tau^{\rep}$, computed from the inflated
weights, and the \emph{true} consistency time $\tau^{\tru}$, computed from the
joint-consistent ones. Though the per-link audit cannot see this inflation,
Theorem~\ref{thm:bias} proves the error is bounded, one-sided, and set by the
worst-audited link.

\begin{theorem}[Consistency-time bias under inflated trust]\label{thm:bias}
For any sheaf $\Fcal'$ and inflations $\beta_e\ge1$, the ratio of the reported
to the true consistency time is bounded by
\begin{equation}
\Big(\max_{e}\beta_e\Big)^{-1}\;\le\;
\frac{\tau^{\rep}(\Fcal')}{\tau^{\tru}(\Fcal')}\;\le\;1 .
\label{eq:bias}
\end{equation}
\end{theorem}

\begin{IEEEproof}
Because $w(\gamma)$ is concave and increasing with $w(0)=0$, the distorted
weights $w(\beta_e\gamma_e)$ are sandwiched as
$w(\gamma_e)\le w(\beta_e\gamma_e)\le(\max_{e'}\beta_{e'})\,w(\gamma_e)$. This
ordering transfers to the Laplacians and hence to $\lambda_2$, and inverting
yields~\eqref{eq:bias}.
\end{IEEEproof}

This theorem provides a predictable fail-safe. Because $\beta_e\ge1$, the
reported time $\tau^{\rep}$ is optimistic, but the error is tightly capped.
\emph{(i) Lower bound:} the network never synchronizes faster than reported,
$\tau^{\rep}\le\tau^{\tru}$. \emph{(ii) Capped delay:} errors do not compound;
the true delay is capped by the single worst-measured link,
$\tau^{\tru}\le\tau^{\rep}\times\max_{e}\beta_e$. \emph{(iii) Guaranteed
budgets:} multiplying the reported time by the maximum known inflation yields
a safe upper bound, so if that bound meets the real-time requirement, the true
system does too. This fail-safe complements the local scaling law:

\begin{corollary}[Locality improves the consistency time]\label{cor:local}
Restricting the check to a ball $B_r(v)$ improves the error bound to
$(\max_{e\subseteq B_r(v)}\beta_e)^{-1}$, rendering the fidelity of the local
consistency time immune to poorly audited links outside the ball.
\end{corollary}

Locality thus offers a double advantage: local neighborhoods not only
reconcile faster (Proposition~\ref{prop:scaling}), their reported consistency
times are also more faithful (Corollary~\ref{cor:local}). A society can
therefore guarantee reliable near-real-time local synchronization even if its
distant periphery degrades.

\section{Numerical Results}\label{sec:results}

We evaluate the framework in two regimes: elicited LLM agent societies, to
test whether the trust shortfall of Theorem~\ref{thm:main} is real, and
simulated sheaves swept across topologies and sizes, to test the bounds of
Proposition~\ref{prop:scaling} and Theorem~\ref{thm:bias}.

\begin{figure}[tb]
  \centering
  \includegraphics[width=1\linewidth]{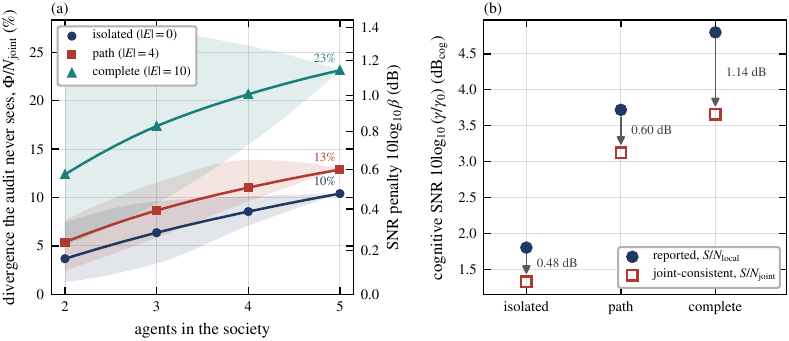}
  \caption{(a)~Unseen divergence share $\Ph/\Njoint$ across network topologies as a function of society size; the right axis carries the same measurement as the SNR penalty $10\log_{10}\beta$. (b)~Reported ($S/\Nloc$) versus joint ($S/\Njoint$) cognitive SNR in $\dBcog$; their separation is that penalty at $n=5$, and is independent of the reference (Corollary~\ref{cor:beta}).}
  \label{fig:identity}
\end{figure}

\subsection{Experimental Setup}\label{ssec:setup}

Our experiment runs five xApps with overlapping E2SM-KPM observation masks
(Fig.~\ref{fig:sys}). All agents track a shared \texttt{MOB} counter
(\texttt{MM.HoExeFailCell}), establishing a baseline correlation, and each is
a locally hosted \texttt{Mistral 7B} model (temperature $0.6$, a single sample
per message). To separate interaction effects from that baseline, identical
scenarios and seeds run across three topologies: (i)~an isolated control
without communication, where correlation stems solely from the shared
counter; (ii)~a path topology along a linear chain; and (iii)~a complete
topology with full peer context.

Each agent reports its own assessment of the cell, using its neighbors'
messages purely as context rather than claims to restate, as
in~\eqref{eq:msg}. The expectations $q_v$ are elicited from the same model
given $o_v$ alone, with no neighbor messages. Generator and auditor therefore
share a prior, which understates $N_v$: every number below is the optimistic
end of what an operator auditing against an independent reference would
measure. To measure asserted precision without altering live communication,
each agent is prompted a second time with a fresh seed, and this second
answer is never transmitted. Pooled across scenarios, the squared difference
between the two responses gives $\sigma^2_{v,c}$, floored at two score points
so that identical answers cannot generate infinite trust; the signal~$S_v$ from~\eqref{eq:signal} is then the sum of the reciprocal variances. Because
every quantity reported below is a ratio in which $S_v$ cancels, none of them
depends on this elicitation: it fixes the absolute SNRs of
Fig.~\ref{fig:identity}b, not their separation.

\begin{figure}[tb]
  \centering
  \includegraphics[width=0.9\linewidth]{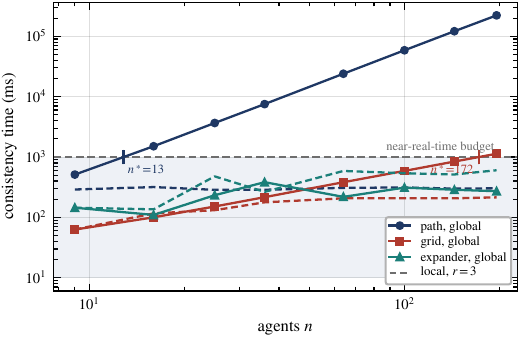}
  \caption{Consistency time versus society size (12 random sheaves per point; 10\,ms per diffusion round). Solid lines denote global consistency times; dashed lines denote local radius-$3$ audits. The shaded band marks the $1$\,s near-real-time budget.}
  \label{fig:scaling}
\end{figure}

\subsection{The Composition Law}

Per-link audits miss joint divergence, creating a blind spot that expands with
network scale and density (Fig.~\ref{fig:identity}a). At $n=5$, hidden
divergence reaches $10\%$, $13\%$ and $23\%$ ($\beta\approx1.12$, $1.15$, and
$1.30$) for the isolated, path, and complete societies. The $10\%$ isolated
baseline forms a structural floor from the shared counter. Communication
expands the gap: a complete topology hides nearly double a path's share, with
$13$ of its $23$ points driven by peer persuasion rather than shared data.

This shortfall distorts the trust layer by inflating reported confidence
(Fig.~\ref{fig:identity}b). Inter-agent context sharpens beliefs and raises
both the reported and the joint cognitive SNR, but the reported one grows
faster. Their separation is the penalty of Corollary~\ref{cor:beta}, widening
from $0.48$\,dB with isolated agents to $0.60$\,dB along a path and $1.14$\,dB
in a complete society. Dense societies thus present a paradox: they report the
highest confidence precisely where the audit harbors its largest blind spot.

\subsection{The Scaling Law}

Sweeping network size ($n=9\dots196$) confirms that topology dictates
real-time feasibility (Fig.~\ref{fig:scaling}), matching
Table~\ref{tab:scaling}: (i)~path topologies scale quadratically,
$\Theta(n^2)$, exceeding the $1$\,s near-real-time budget at $n^\ast=13$ agents
($0.5$\,s to $>200$\,s); (ii)~grids scale linearly, $\Theta(n)$, breaching it
at $n^\ast=172$ ($60$\,ms to $1.1$\,s); and (iii)~expanders maintain flat
$\Theta(1)$ consistency times ($100$--$400$\,ms), remaining within budget at
every size. Local audits (dashed lines) bypass the global limits of sparse
networks: a radius-$3$ audit converges in $\approx300$\,ms on paths regardless
of size, roughly $700\times$ faster than the global consistency time at
$n=196$. On grids, local times plateau near $200$\,ms. On expanders, local and
global consistency times coincide, as Proposition~\ref{prop:scaling} predicts:
local shortcuts are unnecessary when global mixing is already fast.

\subsection{The Consistency-Time Law}

\begin{figure}[tb]
  \centering
  \includegraphics[width=0.9\linewidth]{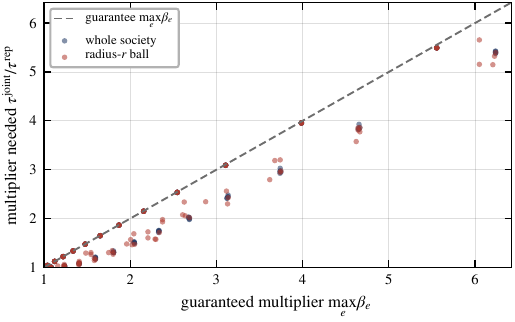}
\caption{The multiplier $\tau^{\tru}/\tau^{\rep}$ needed versus the one Theorem~\ref{thm:bias} guarantees, over 120 random sheaves on $n=100$ agents, each certified globally and on a radius-$3$ ball.}
  \label{fig:clock}
\end{figure}

We evaluate Theorem~\ref{thm:bias} as an operational budgeting rule across
120 random sheaves, comparing the multiplier actually needed,
$\tau^{\tru}/\tau^{\rep}$, against the maximum link inflation $\max_e\beta_e$
(Fig.~\ref{fig:clock}). The bound holds in every trial. Operators can
therefore prevent timing violations by scaling the reported consistency time
by $\max_e\beta_e$, computed from the message streams through $\Ph$.
Specifically, (i)~trials land on the diagonal across scales, showing that the
bound is tight; (ii)~not knowing where the inflation sits costs only minor
slack, under $1.5\times$ (e.g., $3.0\times$ needed against $3.7\times$
guaranteed); and (iii)~radius-$3$ balls obey the same rule with the maximum
taken over the links inside the ball (Corollary~\ref{cor:local}), none of
their points lying above the diagonal. A neighborhood's consistency time can
thus be budgeted from the inflation of its own links alone, whatever the
audit quality elsewhere in the society.

\section{Conclusion}\label{sec:conclusion}

Link-by-link trust evaluation keeps an O-RAN RIC within its tight control loop
but obscures how agent conclusions depend on one another. This blind spot is
exact: the joint divergence exceeds the reported local sum by the total
correlation of the transmitted messages. In agent societies this creates a
trust paradox: as communication increases, reported confidence rises while
unobserved dependence hides up to $23\%$ of the joint divergence in dense
networks. Modeling the control plane as a cellular sheaf shows that this
mis-measured audit still certifies valid consensus states, sacrificing only
convergence speed. Reconciliation decomposes by metric counter, making
consensus as slow as the worst-connected dimension, and consistency time
grows quadratically with network diameter: path topologies breach a $1$\,s
near-real-time budget at $n^\ast=13$ agents and grids at $n^\ast=172$, whereas
radius-$3$ local audits converge in $\approx300$\,ms regardless of network
size. Multiplying the reported consistency time by the worst-link inflation
$\max_e\beta_e$ provides a safe operational bound. Two practical strategies
emerge to secure agentic control planes: (i)~architecturally, by minimizing
network diameter, deploying expanders, or relying on local audits; and
(ii)~informationally, by estimating unobserved message correlations directly
from relayed traffic, without access to proprietary agent internals.

\bibliographystyle{IEEEtran}
\bibliography{refs}

\end{document}